\documentclass[11pt]{article}

\usepackage[T1]{fontenc}
\usepackage{fullpage}
\usepackage{amsmath, amssymb, amsthm, mathtools}
\usepackage{thmtools}
\usepackage[normalem]{ulem}
\usepackage{comment}


\usepackage[colorlinks=true,linkcolor=blue,citecolor=blue,urlcolor=blue]{hyperref}
\usepackage{cleveref}
\usepackage{bm}
\usepackage{tikz}
\usetikzlibrary{angles,arrows.meta,quotes}
\usepackage{tcolorbox}
\usepackage{enumitem}  
\usepackage[margin=1in]{geometry}  

\newtheorem{theorem}{Theorem}[section]
\newtheorem{lemma}[theorem]{Lemma}

\theoremstyle{definition}

\theoremstyle{remark}

\title{The Cube-Root Phenomenon in Online Carpooling}
\author{
Nikhil Bansal\thanks{University of Michigan. \texttt{bansaln@umich.edu}.
Supported in part by NSF awards CCF-2327011 and CCF-2504995.} 
\and
Milind Prabhu\thanks{University of Michigan. \texttt{milindpr@umich.edu}. Supported by NSF award CCF-2327011.} 
\and
Sahil Singla\thanks{School of Computer Science, Georgia Tech. \texttt{ssingla@gatech.edu}. Supported in part by NSF awards CCF-2327010 and CCF-2440113.}
\and
Siddharth M. Sundaram\thanks{School of Computer Science, Georgia Tech. \texttt{ssundaram38@gatech.edu}. Supported in part by NSF awards CCF-2327010 and CCF-2440113.}
}

\newcommand{\R}{\mathbb{R}}

\newcommand{\Gsamp}{G_{\mathrm{samp}}}
\newcommand{\Greedy}{Greedy}
\newcommand{\expanding}{expanding}
\newcommand{\contracting}{contracting}
\newcommand{\numexpanding}{p}

\newcommand{\topsum}[1]{S_{#1}}

\begin{document}
\maketitle

\begin{abstract}
We consider the online carpooling problem, where edges arrive online and must be
oriented immediately while keeping the discrepancy between the indegree and outdegree at each vertex small. 
We prove that the natural \Greedy{} algorithm incurs discrepancy \(O(\min\{T^{1/3},n\})\) after \(T\) arrivals.
This resolves a question of Ajtai et al.~\cite{AjtaiAspnesNaorRabaniSchulmanWaarts98}, who
showed that any deterministic algorithm must incur \(\Omega(\min\{T^{1/3},n\})\) discrepancy, and gave an algorithm with $O(\min\{T^{1/2},n\})$ discrepancy.

\smallskip

We also show a similar square-root to cube-root improvement in the stochastic setting, where $O(n)$ edges are sampled independently from an underlying $n$-vertex graph $G$. Formally, we show an $
O((\log n)^{1/3})$ bound for random arrivals from any $\Delta$-regular graph $G$. When $\Delta = \Omega((\log n)^3)$, we show the more refined bound of $
O((\log n/\log \Delta)^{1/3}+\log\log n)$ on the discrepancy. We show that the cube-root term in the previous bound is essential, while the $\log\log n$ term is already known to be necessary for random arrivals from complete graphs.
The previous upper bounds here were $O(\sqrt{\log n})$, which follow from the breakthrough works on online discrepancy \cite{KulkarniReisRothvoss24,adenali2026optimalonlinediscrepancyminimization}.

\smallskip

Our techniques for proving such cube-root-type bounds may be of
independent interest, as the standard quadratic-potential and
subgaussian analyses underlying the previous general bounds appear
inherently unable to go below square-root-type guarantees.

\end{abstract}

\clearpage

\newpage

\section{Introduction}
We consider the classical carpooling problem introduced by 
Fagin and Williams~\cite{FaginWilliams83}, and defined as follows: We are given a set $V$ of $n$ nodes, and undirected edges arrive online. Upon arrival of an edge $(i,j)$, it has to be oriented immediately as either $i\rightarrow j$ or $j\rightarrow i$. The goal is to minimize the maximum imbalance between the in-degree and out-degree of any node. Formally, if $\delta^{-}_t(i)$ and $\delta^{+}_t(i)$ denote the number of in-edges and out-edges incident to $i$ at time $t$,
define the signed discrepancy of node $i$ as $d_t(i):=\delta^{-}_t(i)-\delta^{+}_t(i)$. The goal is to minimize
\[
    \max_t\max_i |d_t(i)|.
\]
Equivalently, this can be viewed as an online discrepancy minimization problem in $\R^n$, where the incoming vectors $v_t$ have the form $\mathbf{1}(i)-\mathbf{1}(j)$ for $i,j\in [n]$, where $\mathbf{1}(i)$ denotes the $i$-th standard basis vector. The goal is to pick signs $\epsilon_t \in \{-1,1\}$ online to minimize the maximum discrepancy $\max_t \|d_t\|_\infty$, where $d_t = \sum_{s\leq t} \epsilon_{s} v_s$.
We also consider the weighted version of the problem, where each edge arrives with a weight and the goal is to minimize the maximum imbalance between incoming and outgoing weights.

Interestingly, despite its simple structure, the carpooling problem captures several more general problems and phenomena. For example, Ajtai et al. \cite{AjtaiAspnesNaorRabaniSchulmanWaarts98} showed that it is equivalent, up to a factor of $2$,  to {\em general carpooling} where arbitrary size hyperedges (representing carpools) arrive online, from which one person must be designated the driver; the ``fair share'' \cite{FaginWilliams83} for person $i$ is defined as $\sum_{t: i\in e_t} 1/|e_t|$, and we would like to ensure that the number of times any person drives stays close to this. 
In another result, \cite{BRS22} show that the approximability of minimizing the maximum flow time on related machines is, up to $O(1)$ factors, the same as offline edge-weighted carpooling. 

More generally, the problem is also closely related to several other areas such as the power of two-choices \cite{azar1994balanced, kenthapadi2005balancedallocationgraphs}, chip games \cite{anderson-etal, Bls91}, 
fairness in scheduling and social-networks \cite{FaginWilliams83, AjtaiAspnesNaorRabaniSchulmanWaarts98, FiatKarlinKoutsoupiasMathieuZach16}, online discrepancy \cite{Barany1979,BansalSpencer20,AltschulerT2025threshold} and rounding \cite{BRS22, liu_et_al26}, and has received a lot of attention.

However, despite much interest, there remains a long-standing and fundamental {\em cube-root} versus {\em square-root} gap for the problem --- both in the deterministic and the stochastic setting.

\medskip 
{\bf Deterministic Setting.} Here we consider online carpooling where the algorithm is deterministic and the arrivals are adversarial. 
In their seminal work, Ajtai et al.~\cite{AjtaiAspnesNaorRabaniSchulmanWaarts98} showed that the natural \Greedy{} algorithm, that orients each
arriving edge towards the endpoint with smaller discrepancy, achieves an \(O(T^{1/2})\) discrepancy after $T$ requests.\footnote{More precisely $O(\min\{T^{1/2},n\})$.}
On the other hand, they show that any deterministic online algorithm must also incur at least \(\Omega(T^{1/3})\)  discrepancy, provided that $T\leq n^3$.
Coppersmith et al.~\cite{CoppersmithNowickiPaleologoTresserWu11} further showed that \Greedy{} is optimal for this problem.
However, the asymptotic bound on the discrepancy achieved by \Greedy{} had remained unclear.

There are two key conceptual reasons for this gap.
First, both the lower and upper bounds of \cite{AjtaiAspnesNaorRabaniSchulmanWaarts98} are based on tracking the standard quadratic potential $\|d_t\|_2^2$. But using this to control $\|d_t\|_{\infty}$ is inherently lossy and cannot bridge the $T^{1/3}$ versus $T^{1/2}$ gap.
Second, there is a simple $\Omega(T^{1/2})$ lower bound for arbitrary $2$-sparse vectors $v_t$
of the form $a\mathbf{1}(i) + b\mathbf{1}(j)$ with $a,b \in [-1,1]$. 
Thus, any sub-square-root upper bound must exploit the additional
``combinatorial structure'' of carpooling: the two non-zero coordinates
of $v_t$ have equal magnitudes and opposite signs.
We describe these in detail in Section \ref{sec:overview}, as they will be crucial for understanding our ideas.

\medskip 

{\bf Stochastic Setting.}
 Here we assume that there is an underlying graph $G$ and at each time step, the request is an edge of $G$ sampled uniformly at random.
Ajtai et al.~\cite{AjtaiAspnesNaorRabaniSchulmanWaarts98} considered the case where  $G$ is the complete graph $K_n$, and showed that the expected discrepancy is $O(\log \log n)$ at any time; analogous to the power of two-choices phenomenon in load balancing~\cite{azar1994balanced, berenbrink2000balanced}.

The harder setting of general graphs $G$ was studied by Gupta et al.~\cite{GuptaKrishnaswamyKumarSingla20}, who gave an algorithm with polylog($T$) discrepancy. They assume that $G$ is known to the algorithm, and use expander decomposition techniques together with the ideas of \cite{peres2015graphical} to obtain the bound.
Currently, the best known bound for general graphs is $O(\log^{1/2} T)$, which follows from the breakthrough, and more general, results of Kulkarni, Reis, and Rothvoss \cite{KulkarniReisRothvoss24} and  Aden-Ali \cite{adenali2026optimalonlinediscrepancyminimization} on online discrepancy. 

On the other hand, the best lower bound is \(\Omega(\log^{1/3} T)\) \cite{AjtaiAspnesNaorRabaniSchulmanWaarts98}, which follows directly from the deterministic $\Omega(T^{1/3})$ lower bound above: 
roughly, the random process has sufficiently many independent opportunities
that one block of $\Omega(\log T)$ arrivals realizes the equal-discrepancy
choices used by the deterministic lower-bound adversary. Closing this   $\log^{1/3} T$ versus $\log^{1/2} T$ gap 
has been a natural open question, see e.g.,~Conjecture 2 in \cite{KulkarniReisRothvoss24}.
 
\subsection{Our results}
Our first result completely resolves the deterministic case by showing
that the \Greedy{} algorithm itself matches the lower bound of Ajtai
et al.~\cite{AjtaiAspnesNaorRabaniSchulmanWaarts98}.

\begin{theorem}[Deterministic upper bound]\label{thm:greedy}
For any sequence of \(T\) edge arrivals with weights
$w_t\in[0,1]$ on a graph on $n$ vertices, the \Greedy{} algorithm satisfies
\[
    \max_{t\le T}\|d_t\|_\infty
        = O\bigl(\min\{T^{1/3},n\}\bigr).
\]
\end{theorem}

This cube-root upper bound for carpooling is perhaps surprising given the $\Omega(T^{1/2})$ lower bound for 2-sparse vectors, mentioned previously. Indeed, our analysis will crucially exploit that vectors have the form $w_t(\mathbf{1}(i)-\mathbf{1}(j))$.

A key idea is to track a hierarchy of top-\(r\) potentials that captures the entire profile of the discrepancy vector \(d_t\), and use this to prove a refined structural property that we call  {\em spreadness}.
 Roughly, this says that if the discrepancy ever reaches $D$ at some time $t$, then at some previous time $t'\leq t$, there must be a set of $\Omega(D)$ coordinates whose average absolute discrepancy is $\Omega(D)$.
This forces the quadratic potential at that earlier time to be $\Omega(D^3)$; since
\Greedy{}'s quadratic potential is always $O(T)$, the desired
$D=O(T^{1/3})$ bound follows.

Our second result considers the stochastic setting, where we show that a similar cube-root phenomenon holds for regular graphs with $T=O(n)$ requests.

\begin{theorem}[Stochastic upper bound]\label{thm:uniform-rand-reg}
For $T=O(n)$ edges sampled independently from a \(\Delta\)-regular graph \(G\),  there is an  algorithm that orients the edges online such that 
w.h.p.~$\max_{t\leq T}\|d_t\|_\infty
    =O\bigl((\log n)^{1/3}\bigr)$.
Further, for \(\Delta=\Omega(\log^3 n)\) we have the improved bound
\begin{equation}
\label{eq:stoch}
    \max_{t\leq T}\|d_t\|_\infty =
        O\big((\log n/ \log\Delta)^{1/3}
        +\log\log n\big).
\end{equation}
This also holds in the edge-weighted setting, and in fact the weights \(w_t\) may even be chosen by an adaptive adversary based on the
previous samples \(e_1,\ldots,e_{t-1}\).
\end{theorem}
In particular, this 
generalizes the bound
$O(\log \log n)$ of \cite{AjtaiAspnesNaorRabaniSchulmanWaarts98} for complete graphs to arbitrary $n^{\Omega(1)}$-regular graphs. 

We prove a matching lower bound for the cube-root term  in
\eqref{eq:stoch}; the additive
$\log\log n$ term is already known to be necessary for complete graphs.

\begin{theorem}[Stochastic lower bound]
    \label{thm:cube-root-lb}
For every sufficiently large $n$ and every $\Delta=\Omega(\log n)$, there exists a $\Delta$-regular graph $G$ on $n$ vertices\footnote{We consider only parameter pairs $(n,\Delta)$ for which a $\Delta$-regular graph on $n$ vertices exists; equivalently, $\Delta<n$ and $n\Delta$ is even.} such that, for $n$ i.i.d.~requests from $G$, every online algorithm has, with probability $1-o(1)$,
\[
    \max_{t\le n}\|d_t\|_\infty
    =\Omega((\log n/\log \Delta)^{1/3}).
\]
\end{theorem}

The proof of Theorem \ref{thm:uniform-rand-reg} builds on the witness tree approach of Kenthapadi and Panigrahy \cite{kenthapadi2005balancedallocationgraphs}, who generalized the classic power of 2-choices in load balancing from complete graphs to general $\Delta$-regular graphs with $T=O(n)$ arrivals to show an $O((\log n / \log \Delta) + \log \log n)$ bound.
While the $(\log n / \log \Delta)$ bound is optimal for load-balancing, we can get an 
improved cube-root dependence  $(\log n/ \log \Delta)^{1/3}$ in the discrepancy setting by combining this with the ideas in Theorem \ref{thm:greedy}. A direct application of Theorem \ref{thm:greedy} only gives the weaker $O(\log^{1/3} n)$ bound. Obtaining the improved bound in \eqref{eq:stoch} requires opening the black-box and additional new ideas.

{\bf Open question.}
Theorem \ref{thm:uniform-rand-reg} suggests that, for $T$ edges sampled
independently and uniformly from an arbitrary (possibly non-regular) graph $G$,
there is an online algorithm with discrepancy $O(\log^{1/3} T)$. We leave this as an
interesting open question.

\subsection{Overview and Techniques }
\label{sec:overview}
To motivate our approach, we begin by describing the elegant $\Omega(T^{1/3})$ and $O(T^{1/2})$ bounds of \cite{AjtaiAspnesNaorRabaniSchulmanWaarts98} based on the quadratic potential, and see the reasons for this gap.
For simplicity, we assume throughout this section that all edges have weight $1$.

\medskip {\bf Quadratic potential.}
Let $d_t$ denote the discrepancy vector at the end of time $t$. Consider the quadratic potential
$\Phi(t) = \|d_t\|_2^2 = \sum_i d_t(i)^2 $.

The key observation is that for an arriving edge $e_{t}=(i,j)$ (equivalently the vector $v_t = \mathbf{1}(i)-\mathbf{1}(j)$), if $d_{t-1}(i)=d_{t-1}(j)$ at its endpoints, then $\Phi(t)$ must increase by $+2$ irrespective of how $e_t$ is oriented.
On the other hand, if $d_{t-1}(i) \neq d_{t-1}(j)$, choosing the orientation (sign $\epsilon_t$ for $v_t$) greedily cannot increase $\Phi$. This follows as
\begin{align} \Phi(t) - \Phi(t-1) & = 
\|d_{t-1} + \epsilon_t v_t\|_2^2  - \|d_{t-1}\|_2^2 \nonumber \\
&=
2 \epsilon_t ( d_{t-1} \cdot v_t ) + \|v_t\|_2^2 = -2 |d_{t-1}(i)-d_{t-1}(j)| + 2. \label{eq:lbpot}
\end{align}

\medskip

{\bf Previous bounds and limitations.}
As $\Phi(t)$ increases by at most \(+2\) per request and $\Phi(0)=0$, we have \(\Phi(t)\le 2t\le 2T\) trivially, giving the upper bound $\|d_t\|_\infty =O(\sqrt{T})$ for each \(1\leq t\leq T\).

For the lower bound, fix some time $T \leq n^3$, and let $h=T^{1/3}$. We will only request edges $(i,j)$ with $i,j\in [h]$. We have a win-win argument.
If $\|d_t\|_\infty \geq h/2$ for some $t\leq T$, we already have the $\Omega(h)$ lower bound. Else, at each time $t$ all discrepancies lie in $(-h/2,h/2)$ and by pigeonhole, there must exist some $i,j \in [h]$ with $d_{t-1}(i)=d_{t-1}(j)$. Requesting $e_t=(i,j)$ thus forces $\Phi$ to rise by $+2$,  eventually giving $\Phi(T)=\|d_T\|_2^2 = 2T = 2h^3$. Since $d_T$ is supported only on vertices in $[h]$,  some vertex must have $\Omega(h)$ discrepancy.

 Intuitively, the problem is that using a bound on $\|d_t\|_2^2$ to infer $\|d_t\|_\infty$ is inherently lossy, as one cannot distinguish whether the contribution to $\|d_t\|_2^2$  is coming from one large coordinate, or multiple large coordinates. This  leads to the cube-root versus square-root gap.

\medskip
{\bf Lower bound for 2-sparse vectors.}
Second, an analysis that treats carpooling requests merely as arbitrary bounded $2$-sparse vectors cannot yield an $o(T^{1/2})$ bound.
Indeed, there is a trivial $\Omega(T^{1/2})$ lower bound for such vectors, already when $n=2$ ---
at each time $t$, we can always choose some unit length vector $v_t$ supported on the first two coordinates that is orthogonal to $d_{t-1}$.
By orthogonality, for either choice of sign,
\[
\|d_{t-1}+\epsilon_t v_t\|_2^2
 =\|d_{t-1}\|_2^2+\|v_t\|_2^2
 =\|d_{t-1}\|_2^2+1.
\]
Thus $\Phi(t)=t$, and hence $\|d_t\|_\infty=\Omega(t^{1/2})$ as its support has size $n=2$.

\medskip
{\bf Deterministic Arrivals and Spreadness.}
Our main insight is that under the \Greedy{} algorithm, there must be times when the $\ell_2$-mass in $d_t$ must be spread among several coordinates.
In particular, if the discrepancy $\|d_T\|_\infty = D$ at time $T$, then at some earlier time $t\le T$ there is a set $R$ of $\Omega(D)$ coordinates whose average absolute discrepancy is $\Omega(D)$. We call this the spreadness property. It implies,  by Cauchy--Schwarz, that
\[
    \Phi(t)=\sum_i d_t(i)^2
    \geq \frac{1}{|R|}\left(\sum_{i\in R}|d_t(i)|\right)^2
    =\Omega(D^3).
\]
As $\Phi(t)\le 2t\le 2T$, this yields
$D=O(T^{1/3})$.

To show that such a set $R$ must exist at some earlier time, we simultaneously track, for every $r$, the top-$r$ sum $S_r(t)$, defined as the sum of the $r$ largest signed discrepancies at time $t$.
However, one serious difficulty is that the mass need not be spread at all at any given time. E.g., even if the mass is currently spread, the adversary can give new requests to bring all coordinates back close to $0$, except the one with the largest value (leaving the discrepancy unchanged).

To handle this we prove a stronger inductive property than spreadness.  We track the evolution of the $S_r(t)$ over time. We show that if a particle reaches high discrepancy at time $T$, then, for every $r$, there is a time $t_r\leq T$ at which $S_r(t_r)$ crosses a suitably-chosen threshold. The existence of such times where $S_r$ is large is proved via a careful backward induction over time.

\medskip
{\bf Stochastic arrivals.} We now describe the key ideas for Theorem \ref{thm:uniform-rand-reg}.
By dividing the $T=O(n)$ requests into $O(1)$ batches of size $cn$  for $c \ll 1$, it suffices to prove the bound for such a batch. Let \(G_{\mathrm{samp}}\) denote the \textit{sampled multigraph},
 formed by the \(cn\) sampled requests. A standard witness tree
argument shows that w.h.p.~every connected component $C$ of
\(G_{\mathrm{samp}}\) has \(O(\log n)\) vertices and \(O(\log n)\) edges.
As different components do not interact, one can essentially apply the bound on \Greedy{} from Theorem \ref{thm:greedy} to each component individually, which immediately gives the
\(O((\log n)^{1/3})\) discrepancy bound.

\medskip
{\bf  The stronger bound \eqref{eq:stoch}.} Proving the stronger bound given by \eqref{eq:stoch} is more subtle, and requires two additional key ideas.

First, adapting the witness-tree framework of Kenthapadi and Panigrahy
\cite{kenthapadi2005balancedallocationgraphs}, we show that when \(\Delta\) is
large, every connected component \(C\) of \(G_{\mathrm{samp}}\) is nearly a
tree. In particular, if we view $C$ as consisting of a spanning tree with $m-1$ edges plus $p$ surplus internal edges, then w.h.p. we have $m=O(\log n)$ and $p = O(\log n/\log \Delta)$.

Second, we exploit this structure and prove a stronger $O(p^{1/3}+\log m)$ bound for such tree-like adversarial request sequences, in contrast to the general \(O((m+p)^{1/3})\) bound from Theorem \ref{thm:greedy}. Notice that the cube-root dependence now is only on the surplus edges $p$, while the dependence on $m$ is logarithmic. 

To prove this refined bound, we analyze the effect of edge arrivals differently depending on whether they merge two different components, or are already inside some component.
To show that only the internal edges incur the cube-root dependence, a key idea is to consider another auxiliary carpooling instance with one extra vertex, where the discrepancies of the original vertices are shifted by the logarithm of the current size of the component containing them. We can then use the analysis for Theorem \ref{thm:greedy} on this instance and infer the refined bound. We refer to Section \ref{sec:stochastic-arrivals} for details.

\section{Preliminaries}
\label{sec:preliminaries}

We begin by defining the problem formally, using the discrepancy formulation.

 There are \(n\) coordinates indexed by $[n]$ and an unknown time horizon $T$. At each time $t \in [T]$, a request arrives online, specified by a weight $w_t \in [0,1]$ and two distinct coordinates $i_t,j_t \in [n]$, and the algorithm must pick a sign $\epsilon_t \in \{-1,1\}$.
Let $v_t$ denote the vector $w_t(\mathbf{1}(i_t)-\mathbf{1}(j_t))$, and  $d_t = \sum_{s \leq t} \epsilon_s v_s$ the discrepancy vector at the end of time $t$.
The goal is to minimize $D := \max_{t} \|d_t\|_\infty$.

\medskip
{\bf The \Greedy{} algorithm and chip view.}
We will only consider the \Greedy{} algorithm, which, upon the request at time \(t\), if
\(d_{t-1}(i_t)\le d_{t-1}(j_t)\), chooses $\epsilon_t$ so that:
\begin{equation}
\label{eq:greedy-update}
    d_t(i_t)=d_{t-1}(i_t)+w_t
    \qquad\text{and}\qquad
    d_t(j_t)=d_{t-1}(j_t)-w_t.
\end{equation}
It will be convenient to view the discrepancies as \(n\) chips on the real line, one for each coordinate, where each chip $i$ is at position $d_t(i)$ at the
end of time \(t\).
In this {\em chip view}, each request picks two chips, and \Greedy{} moves the left chip by $w_t$ to the right, and the right chip by $w_t$ to the left.

The question is: how far can an adversary push
any chip from the origin?

\medskip{\bf Expanding and contracting moves.}
For the purpose of our analysis, we will distinguish two
types of chip moves by \Greedy{}. Consider time \(t\) and let $i,j$ denote the requested chips so that
\(d_{t-1}(i)\le d_{t-1}(j)\). A move is \emph{\expanding} if
\(d_t(i)>d_{t-1}(j)\), and \emph{\contracting} otherwise. In words, a move is
\expanding\ if chip $i$ moves to the right of chip $j$'s previous position.
Equivalently, an \expanding\ move increases the distance between the two
chips $i$ and $j$.
See Figure \ref{fig:chip-game-moves}.
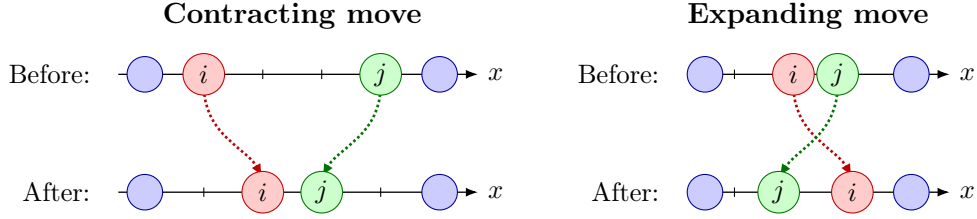
\begin{figure}[ht]
\centering
\begin{tikzpicture}[
    x=0.78cm,
    y=0.78cm,
    >=Latex,
    axis/.style={-{Latex[length=1.6mm]},line width=.5pt},
    moveu/.style={-{Latex[length=1.8mm]},densely dotted,line width=1pt,red!75!black},
    movev/.style={-{Latex[length=1.8mm]},densely dotted,line width=1pt,green!45!black},
    chipu/.style={circle,draw=red!75!black,fill=red!20,minimum size=5.5mm,inner sep=0pt},
    chipv/.style={circle,draw=green!45!black,fill=green!20,minimum size=5.5mm,inner sep=0pt},
    otherchip/.style={circle,draw=blue!60!black,fill=blue!20,minimum size=4.7mm,inner sep=0pt},
    every node/.style={font=\small}
]

\begin{scope}
    \node[font=\bfseries] at (1.5,3.0) {Contracting move};
    \node[left] at (-1.75,2) {Before:};
    \node[left] at (-1.75,0) {After:};
    \draw[axis] (-1.45,2) -- (4.65,2) node[right] {$x$};
    \draw[axis] (-1.45,0) -- (4.65,0) node[right] {$x$};
    \foreach \x in {-1,0,1,2,3,4}{
        \draw (\x,1.92) -- (\x,2.08);
        \draw (\x,-.08) -- (\x,.08);
    }
    \draw[moveu] (0,1.75) .. controls (.05,.95) and (.72,.75) .. (1,.25);
    \draw[movev] (3,1.75) .. controls (2.95,.95) and (2.28,.75) .. (2,.25);
    \node[chipu] at (0,2) {$i$};
    \node[chipv] at (3,2) {$j$};
    \node[otherchip] at (-1,2) {};
    \node[otherchip] at (4,2) {};
    \node[chipu] at (1,0) {$i$};
    \node[chipv] at (2,0) {$j$};
    \node[otherchip] at (-1,0) {};
    \node[otherchip] at (4,0) {};
\end{scope}

\begin{scope}[xshift=7.8cm]
    \node[font=\bfseries] at (.25,3.0) {Expanding move};
    \node[left] at (-2.1,2) {Before:};
    \node[left] at (-2.1,0) {After:};
    \draw[axis] (-1.8,2) -- (2.65,2) node[right] {$x$};
    \draw[axis] (-1.8,0) -- (2.65,0) node[right] {$x$};
    \foreach \x in {-1,0,1,2}{
        \draw (\x,1.92) -- (\x,2.08);
        \draw (\x,-.08) -- (\x,.08);
    }
    \draw[moveu] (0,1.75) .. controls (.05,.95) and (.72,.75) .. (1,.25);
    \draw[movev] (.75,1.75) .. controls (.70,.95) and (.08,.75) .. (-.25,.25);
    \node[chipu] at (0,2) {$i$};
    \node[chipv] at (.75,2) {$j$};
    \node[otherchip] at (-1.5,2) {};
    \node[otherchip] at (2,2) {};
    \node[chipu] at (1,0) {$i$};
    \node[chipv] at (-.25,0) {$j$};
    \node[otherchip] at (-1.5,0) {};
    \node[otherchip] at (2,0) {};
\end{scope}
\end{tikzpicture}
\caption{Expanding moves increase the distance between the requested chips,
whereas contracting moves do not. The unrequested chips (marked blue) do not
move.}
\label{fig:chip-game-moves}
\end{figure}

\section{Deterministic Setting}
\label{sec:chip-game}
In this section, we prove the bound on \Greedy{} stated in Theorem \ref{thm:greedy}.

In fact, we will show the following slightly stronger result, which allows the starting positions to be negative, and bounds the positive discrepancy in  terms of expanding moves. This form will be useful later in Section \ref{sec:stochastic-arrivals}. 







\begin{lemma}
\label{lem:chip-game}
Suppose that initially $d_0(i)\le0$ for every $i\in[n]$, and the request sequence leads to $p$ expanding moves. Then for every chip $i$ and every time $t$, we have $   d_t(i)\le \min(4p^{1/3},n)$.
\end{lemma}

By symmetry, this directly implies Theorem \ref{thm:greedy} when  all chips start at $0$.

The proof of Lemma \ref{lem:chip-game} has two steps. First, we consider a natural quadratic potential, and show that it is $O(p)$ at any time. Second, we show that if some chip reaches discrepancy $D>0$, then at some {\em earlier time} there is a set of $\Omega(D)$ chips whose average discrepancy is $\Omega(D)$. Combining these, we obtain the $O(p^{1/3})$ bound. This second statement is the main point of the argument.



\subsection{Potential and Spreadness}

For $x\in\R$, write $x_+:=\max\{x,0\}$. 
We consider the following potential for the positive discrepancies.
\[
    \Phi_+(t):=\sum_{i=1}^n (d_t(i)_+)^2.
\]
Note that we only sum over chips at non-negative positions. 
We have the following bound.
\begin{lemma}
\label{lem:energy}
If at most $p$ expanding moves have occurred by time $t$, then $\Phi_+(t)\le2p$.
\end{lemma}

\begin{proof}
Define the function $f(x)=x_+^2$, and note that $f$ is convex. 

Consider a move of weight $w$ from positions $a\le b$ to
$a+w$ and $b-w$.
If the move is contracting, then both $a+w$ and $b-w$ lie in $[a,b]$, and the potential can only decrease by the convexity of $f$.

We now show that each expanding move can increase $\Phi_+$ by at most $2$.

For all $x,y\in\R$ we have the inequality, $f(x+y)\le f(x)+2x_+y+y^2$.
Indeed, for $x\ge0$ the right side is $(x+y)^2$. For $x<0$ it is $y^2$, while the left side is at most $y^2$. 
Applying this with $(x,y) = (a,w)$ and with $(b,-w)$, and using that $a_+\le b_+$, $w\in [0,1]$
gives
\[
    \Phi_+(t)-\Phi_+(t-1) \le 2w(a_+-b_+)+2w^2 
    \le 2.
\qedhere \]
\end{proof}

\paragraph{Top sums and Spreadness.}
For a subset of chips $R\subseteq[n]$, let us denote $d_t(R):=\sum_{i\in R}d_t(i)$.

For $1 \leq r \leq n$, let 
    $S_r(t) := \max_{|R|=r} d_t(R)$
denote the sum of the $r$ largest chip positions at time $t$. For $r=0$, we use the convention $S_0(t):=0$.
Note that $S_r(t)$ can include chips at negative positions.

We will show the following key fact. 
\begin{lemma}[Spreadness lemma]\label{lem:chip-staircase}
Suppose that \(d_0(i)\le0\) for every \(i\in[n]\),
and that the largest chip position at time
$T$ is $D>0$. Then, for every integer $
    1\le r\le \min\{\lfloor D\rfloor,n\}$,
there is a time $t_r\le T$ such that
\[
    \topsum{r}(t_r)\ge r(D-r+1).
\]
\end{lemma}
In particular, setting $r=\lfloor D/2 \rfloor$, this says that there is some time with $\Omega(D)$ chips at positions $\Omega(D)$.
Let us see how Lemma \ref{lem:chip-staircase} directly implies Lemma~\ref{lem:chip-game}.

\begin{proof}[{\bf Proof of Lemma~\ref{lem:chip-game}}]
Let $D:=\max_{t,i}d_t(i)$, and let $T^\star$ be a time at which some chip is at position $D$.
We assume that $D>0$ and is sufficiently large, otherwise the lemma trivially holds as $p\geq 1$.

We first show that \(D\leq n\). Suppose otherwise. Then, applying
\Cref{lem:chip-staircase} with \(r=n\), there is a time \(t_n\le T^\star\)
such that
\[
    \topsum{n}(t_n)\ge n(D-n+1)>0.
\]
But as \(\topsum{n}(t)\) is the sum of all chip positions, and every
move preserves this sum, we have that
    $\topsum{n}(t_n)=\topsum{n}(0)\le0$ (as all chips are initially at nonpositive positions), which is a contradiction.

 We next show \(D=O(\numexpanding^{1/3})\). Let  \(r=\lfloor D/2\rfloor\), and note that $1 \leq r \leq n $ as $D$ is sufficiently large and at most $n$. By Lemma \ref{lem:chip-staircase}, there is a time \(t_r\le T^\star\) for which
\begin{equation}\label{eq:chip-staircase}
    \topsum{r}(t_r)\ge r(D-r+1)\ge D^2/8 >0.
\end{equation}
Let \(R\) be the \(r\) chips with largest positions at time \(t_r\). By Jensen's inequality for $f(x)=x_+^2$, \eqref{eq:chip-staircase} gives
\[
    \Phi_+(t_r)
    \ge \sum_{i\in R}f(d_{t_r}(i))
    \ge r f\left(\frac{1}{r}\sum_{i\in R}d_{t_r}(i)\right)
     =\frac{1}{r} (S_r(t_r)_+)^2 
    \ge \frac{D^3}{32}.
\]
Thus by Lemma~\ref{lem:energy} we have  $D^3/32\le2p$, and hence $D\le4p^{1/3}$.
\end{proof}

\subsection{Proof of the Spreadness Lemma}

We now prove the spreadness lemma.

We first need a claim about the evolution of $S_r(t)$.
It says that
if the top \(r\)-chip sum increases by the move at time $t$, we can infer a lower bound on the top \((r+1)\)-chip sum at time $t-1$.

\begin{lemma}
\label{lem:top-sum}
If $1\leq r<n$ and $S_r(t)>S_r(t-1)$, then
 $S_{r+1}(t-1)\ge 2S_r(t)-S_{r-1}(t)-2.$
\end{lemma}


\begin{proof}
Let $u,v$ be the moved chips at time $t$, and let $d_{t-1}(u)=a\le b=d_{t-1}(v)$. So $u$
moves right and $v$ moves left. Let $R$ be some $r$-set attaining $d_t(R) = S_r(t)$ at time $t$. Since,
\[
    d_t(R)=S_r(t)>S_r(t-1)\ge d_{t-1}(R),
\]
it must be that $u\in R$ and $v\notin R$. 
As  $S_{r-1}(t) 
\geq d_t(R\setminus\{u\}) = d_t(R)-d_t(u)$, we have $d_t(u) \ge S_r(t)-S_{r-1}(t)$, and thus
\[
    d_t(v)\ge S_r(t)-S_{r-1}(t)-2,
\]
as $d_t(u)-d_t(v)=a-b+2w\le2$.

Consider the $(r+1)$-set $R\cup\{v\}$. As it contains both $u,v$, its total is unchanged by the move. Thus
\[
\begin{aligned}
    S_{r+1}(t-1)
    &\ge d_{t-1}(R\cup\{v\})
     =d_t(R\cup\{v\})
    \ge 2S_r(t)-S_{r-1}(t)-2.
\end{aligned} \qedhere
\]
\end{proof}

We can now prove Lemma \ref{lem:chip-staircase}.

\begin{proof}[{\bf Proof of \Cref{lem:chip-staircase}}]
Let $m:=\min\{\lfloor D\rfloor,n\}$. 
For \(0\le r\le m\), define the threshold $A_r:=r(D-r+1)$. For $r\ge1$, let $t_r$ be the
{\em first} time at which $S_r(t_r)\ge A_r$, provided such a time exists. 

We prove by induction on $r \geq 1$ that
all these times exist and that
$ T\ge t_1>t_2>\cdots>t_m$.

The base case of $r=1$ is immediate as
\(\topsum{1}(T)=D=A_1\), and thus \(t_1\) exists and satisfies \(t_1\le T\).

Suppose that $1\le r<m$ and that
$t_1>\cdots>t_r$ have been defined. As $A_r>0$ for $r \leq D$ and as $S_r(0)\le 0$ (by initial chip positions), we have that $t_r\ge1$.
Moreover, by minimality of $t_r$, we have 
$S_r(t_r-1)<A_r\le S_r(t_r)$.
Thus $S_r$ increases at time $t_r$, and by Lemma~\ref{lem:top-sum} we have 
\begin{equation}
\label{eq:rel1}
S_{r+1}(t_r-1)\ge2S_r(t_r)-S_{r-1}(t_r)-2.
\end{equation}
As 
$t_r<t_{r-1}$, the
minimality of $t_{r-1}$ again gives $S_{r-1}(t_r)<A_{r-1}$ (for $r=1$, we use $S_0(t)=A_0=0$). Thus by \eqref{eq:rel1},
\[
\begin{aligned}
    S_{r+1}(t_r-1)
    &\ge2A_r-A_{r-1}-2  =(r+1)(D-r)=A_{r+1}.
\end{aligned}
\]
Consequently $t_{r+1}$ exists and $t_{r+1}\le t_r-1<t_r$, which completes the induction.
\end{proof}

\section{Stochastic Arrivals from Regular Graphs}
\label{sec:stochastic-arrivals}
We now consider the stochastic setting and prove Theorem \ref{thm:uniform-rand-reg}.

Recall that here we are given a
\(\Delta\)-regular \emph{base graph} \(G=([n],E)\), and at each time \(t\), the request is an edge
\(e_t=(u_t,v_t)\), sampled independently and uniformly from \(E\).
For ease of exposition, we assume that
each request has weight \(1\), and we use the words request and edge interchangeably.
However, 
the bound in Theorem \ref{thm:uniform-rand-reg} holds for arbitrary weights \(w_t\in[0,1]\), even when $w_t$ is chosen adaptively based on the
 previously sampled edges \(e_1,\ldots,e_{t-1}\), as we will reduce the analysis of the stochastic setting to a suitable deterministic instance and apply Theorem \ref{thm:greedy}.

\medskip
{\bf Sampled Graph.}
As discussed in Section \ref{sec:overview}, we can assume that $T=cn$, and we set $c=1/16$.
The key object in the proof of \Cref{thm:uniform-rand-reg} is the
\emph{sampled multigraph} \(\Gsamp\), which has vertex set \([n]\), and
edge-multiset 
\(e_1,\ldots,e_{cn}\) consisting of edges sampled independently from $G$ (with multiplicity).

The proof has two ingredients: First, we show that
\(\Gsamp\) has connected components of size $O(\log n)$, and that each component is close to a tree. As \Greedy{} works
 on each component of \(\Gsamp\)  (formed in hindsight) separately,  we can apply the adversarial bound in Theorem \ref{thm:greedy} to the request sequence within each component. This already gives the $O(\log^{1/3} n)$ bound in Theorem \ref{thm:uniform-rand-reg}. Second, to prove the stronger bound \eqref{eq:stoch} we prove a more refined adversarial bound for \Greedy{} when the
multigraph formed by the requests is tree-like. 

We now
state these two ingredients formally, and show they imply Theorem \ref{thm:uniform-rand-reg}.

\begin{lemma}\label{lem:rand-graph-lem}
With probability $1-n^{-\Omega(1)}$,
every connected component \(C\) of \(\Gsamp\) satisfies:

 \,\,\,\,\, 1.  {\em Few vertices and edges:}  $C$ has
    $|V(C)| = O(\log n)$ vertices and  $|E(C)|  = O(\log n)$ edges.
    
  \,\,\,\,\, 2. {\em ``Tree-like'' components:} 
    For $\Delta = \Omega(\log^3 n)$, we further have
    $
    |E(C)| \leq |V(C)| + O(\log n/\log \Delta).$
\end{lemma}

In particular, part 2 of Lemma \ref{lem:rand-graph-lem} says 
that each component has only
\(O(\log n/\log\Delta)\) edges beyond a spanning tree. The following lemma gives a refined bound for precisely such request
sequences.

\begin{restatable}[\Greedy{} on  ``tree-like'' requests]
    {lemma}{treelikegreedy}\label{lem:graph-ub}
Let \(H\) be a connected multigraph on \(m\) vertices with \(T\) edges, and
let \(p:=T-(m-1)\) denote the number of surplus edges in $H$
beyond a spanning tree. Let \(\rho_1,\ldots,\rho_T\) be any ordering of the
edges of \(H\), viewed as carpooling requests. Then, for each \(1\leq t\leq T\),
the maximum discrepancy of \Greedy{} after processing \(\rho_1,\ldots,\rho_t\) is
\[
    O\bigl(p^{1/3} + \log m \bigr).
\]
\end{restatable}
\begin{proof}[Proof of \Cref{thm:uniform-rand-reg}]
Condition on the event in \Cref{lem:rand-graph-lem}, and fix a component
\(C\) of \(\Gsamp\). Now \Greedy{} works separately in each component, and the
request sequence in \(C\) is some arbitrary ordering of the sampled edges in $C$.
 As \(|E(C)|=O(\log n)\) by Lemma \ref{lem:rand-graph-lem}, \Cref{thm:greedy} directly gives the
    \(O(\log^{1/3} n)\) bound in Theorem \ref{thm:uniform-rand-reg}.

We now show the second bound in Theorem \ref{thm:uniform-rand-reg} assuming \(\Delta=\Omega(\log^3 n)\). Fix some component $C$ of \(\Gsamp\) and set
    \(m:=|V(C)|\) and \(p:=|E(C)|-|V(C)|+1\).  Conditioning on the event in Lemma \ref{lem:rand-graph-lem}, we have \(m=O(\log n)\) and
    \(p=O(\log n/\log\Delta)\).
    Hence, applying \Cref{lem:graph-ub} to the request sequence in $C$ gives the claimed discrepancy $O(p^{1/3} + \log m)=
        O((\log n/\log\Delta)^{1/3}
            +\log\log n)$.

Since the event in \Cref{lem:rand-graph-lem} holds with probability
\(1-n^{-\Omega(1)}\) and $C$ is arbitrary, with high probability, the claimed discrepancy bound holds  for every $i \in [n]$.
\end{proof}

The proof of \Cref{lem:rand-graph-lem} is similar to the witness-tree argument in~\cite{kenthapadi2005balancedallocationgraphs}.
But the
formulation we need here is slightly different, we defer the details to~\Cref{subsec:structure-sampled} for completeness. We prove \Cref{lem:graph-ub} next.

\subsection{\Greedy{} on Tree-Like Request Sequences}
\label{subsec:treelike}

\begin{proof}[Proof of \Cref{lem:graph-ub}]
For $0\le t\le T$, let $H_t$ be the multigraph formed by the first $t$ requests, and let $c_t(u)$ denote the size of the component containing $u$ in $H_t$. We call a request $\rho_t=(u,v)$ \emph{merging} if $u$ and $v$ lie in different components of $H_{t-1}$, and \emph{internal} otherwise. Since $H$ is connected, exactly $m-1$ requests are merging, and hence exactly $p$ are internal.

The idea is to absorb the effect of merging requests into the growth of the component sizes, so that only the $p$ internal requests can contribute expanding moves. Define
\[
    x_t(u):=d_t(u)-2\log_2 c_t(u),
    \qquad
    y_t(u):=-d_t(u)-2\log_2 c_t(u).
\]
We show that   $x_t(u)\le 4p^{1/3}$
for every $u$ and $t$. The same argument applied to $y_t$ gives the analogous bound. Since $c_t(u)\le m$, it then follows that
\[
    |d_t(u)|\le 4p^{1/3}+2\log_2 m.
\]
To bound $x_t(u)$, consider an auxiliary chip game with one chip for each vertex $u$ and one dummy chip $z$, initially placed sufficiently far to the left of the origin. All the vertex chips start at the origin.
The idea is to realize the evolution of $x_t$  so that every internal request in $H_t$ is simulated by a single chip move in the auxiliary game, while every merging request is simulated using only contracting moves. The auxiliary game then has at most $p$ expanding moves, and Lemma \ref{lem:chip-game} gives the desired bound.

Suppose first that $\rho_t=(u,v)$ is internal, and label its endpoints so that
$d_{t-1}(u)\le d_{t-1}(v)$. Since $u$ and $v$ are in the same component,
$c_{t-1}(u)=c_{t-1}(v)$, and hence
$x_{t-1}(u)\le x_{t-1}(v)$. Moreover, the component sizes do not change, so
\[
    x_t(u)=x_{t-1}(u)+1,
    \qquad
    x_t(v)=x_{t-1}(v)-1.
\]
Thus the update is exactly one chip move of weight $1$.

Now suppose that the request $\rho_t=(u,v)$ merges components $A$ and $B$ containing $u$ and $v$ respectively. Let $a=|A|$ and $b=|B|$, and suppose $d_{t-1}(u)\le d_{t-1}(v)$.
As $d_{t-1}(u)\le d_{t-1}(v)$, we have $d_t(u) = d_{t-1}(u)+1$ and $d_t(v)= d_{t-1}(v)-1$, and thus
\[
    x_t(u)=x_{t-1}(u)+1-\alpha,
    \qquad
    x_t(v)=x_{t-1}(v)-1-\beta,
\]
where 
    $\alpha:=2\log_2((a+b)/a)$ and 
    $\beta:=2\log_2((a+b)/b)$. Note that both $\alpha,\beta > 0$.

Moreover, every chip in $A\setminus\{u\}$ decreases by $\alpha$, and every chip in $B\setminus\{v\}$ decreases by $\beta$.

The dummy chip $z$ serves simply as a sink for these decreases. By placing it sufficiently far to the left, we can decrease any vertex chip by an arbitrary amount using contracting moves with $z$, splitting the decrease into moves of weight at most $1$. Thus all updates for vertices other than $u$ and $v$ are immediate.

It remains to handle $u$ and $v$. If $\alpha\ge1$, both chips move left (note that $v$ always moves to the left as $\beta>0$), and we again use $z$. Suppose therefore that $\alpha<1$, and let
$ w:=1-\alpha\in(0,1)$.
So chip $u$ moves right by $w$. 

A key observation is that this move is contracting. To this end we will show that $x_{t-1}(v)-x_{t-1}(u)\geq w$. Indeed,
\[
\begin{aligned}
    x_{t-1}(v)-x_{t-1}(u)
    =
    d_{t-1}(v)-d_{t-1}(u)
      +2\log_2 (a/b) \ge
    2\log_2 (a/b)
    =\beta-\alpha.
\end{aligned}
\]
Now $\alpha<1$ implies that $(a+b)/a \leq \sqrt{2}$ and hence that $a>b$. But this implies that $\beta>2$, and thus $\beta -\alpha \geq 2 -\alpha  >1-\alpha=w$.

We first make a weight-$w$ move on the chips $u,v$, which is contracting. This moves $u$ to its required position $x_t(u)$  and decreases $v$ by $w$. The remaining decrease required at $v$ is $
    1+\beta-w=\alpha+\beta>0$,
which can again be carried out using $z$.

Thus every merging request can be simulated entirely by contracting moves, whereas each internal request contributes at most one expanding move. Since there are exactly $p$ internal requests, Lemma \ref{lem:chip-game} gives 
    $x_t(u)\le 4p^{1/3}$
for every $u$ and $t$. Applying the same argument to $y_t$ completes the proof.
\end{proof}

\subsection{Lower Bound for Stochastic Arrivals}

We now prove Theorem \ref{thm:cube-root-lb} and show a lower bound for the cube-root bound in \eqref{eq:stoch}. 

\begin{theorem}\label{thm:lower-bound}
Suppose that $\Delta=\Omega(\log n)$. There exists a $\Delta$-regular graph
$G$ on $n$ vertices such that for $n$ i.i.d. edge requests from $G$,
every online algorithm has, with
probability $1-o(1)$,
\[
    \max_{t\le n}\|d_t\|_\infty
    =\Omega\big((\log n/\log\Delta)^{1/3}\big).
\]
\end{theorem}
The proof is an adaptation of the $\Omega(\log^{1/3}T)$ randomized lower bound
of \cite{AjtaiAspnesNaorRabaniSchulmanWaarts98} in Section \ref{sec:overview}.

\begin{proof}[Proof of \Cref{thm:lower-bound}]
Let $G$ be the $\Delta$-regular graph consisting of
$M:=n/(\Delta+1)$ disjoint cliques $C_1, \ldots, C_M$, each of size $\Delta+1$. We will show an $\Omega(h)$ lower bound, where
$h:=\left\lfloor(\log n/(8\log\Delta))^{1/3}\right\rfloor$.

We can assume that $\Delta< n^{1/64}$, otherwise the lower bound holds trivially. Thus $h\ge 2$, and the number of cliques is
$M=\Omega(n^{63/64})$.
We will show that each clique $C_i$ contains a vertex of discrepancy
$\Omega(h)$ with probability $\Omega(n^{-1/4})$, and that these events are independent for different cliques. As there are
$\Omega(n^{63/64})$ cliques, this will give the result.

It is convenient to work with a Poissonized edge arrival process rather than assume that the number of arrivals is exactly $n$. Assign to each edge $e\in E(G)$ an independent Poisson clock of rate $1/\Delta$, and run the process for one unit of time. Since $|E(G)|=n\Delta/2$, the total number $N$ of edges that arrive satisfies $N\sim\operatorname{Poi}(n/2)$. Thus, this yields a sequence of $N$ edges, where each is sampled independently and uniformly from $E(G)$. Finally, $\Pr[N>n]=e^{-\Omega(n)}$, so $N\le n$ with high probability. Therefore, a lower bound in the Poissonized model translates to a lower bound in the model with exactly $n$ arrivals.

Fix a clique \(C_i\) and note that the number of arrivals in \(C_i\) is Poisson with mean \(\frac{1}{\Delta}\binom{\Delta+1}{2}=(\Delta+1)/2\). For $L = h^3 $ which is $o(\Delta)$, by standard concentration bounds, the probability that $C_i$ receives at least $L$ arrivals is at least $1/2$.

Now, fix an arbitrary set \(S_i\subseteq C_i\) of \(h\) vertices and consider the first $L$ arrivals in $C_i$.  If at any point some vertex of \(S_i\) has absolute discrepancy at least \(h/4\), we are done. Otherwise, some two vertices have the same discrepancy. For each $k\le L$, before the $k$-th edge in $C_i$ is revealed, prescribe an edge $e_{i,k}\subseteq S_i$ joining two equal-discrepancy vertices, choosing lexicographically; if the lower bound has already occurred, choose arbitrarily. Let $B_i$ be the event that $C_i$ receives at least $L$ arrivals and each of its first $L$ arrivals equals its prescribed edge.

If the event \(B_i\) occurs, then either a discrepancy of \(\Omega(h)\) has already occurred, or every one of the first \(L\) arrivals increases the potential
    $\Phi_i:=\sum_{v\in S_i} d(v)^2$
by \(2\). In the latter case, after \(L\) arrivals \(\Phi_i=2L=2h^3\), and some vertex has discrepancy \(\Omega(h)\).

As each arrival in \(C_i\) is uniform among its (at most $\Delta^2$) edges, 
$ \Pr[B_i]
    \ge (1/2) \Delta^{-2h^3}
    \ge(1/2) n^{-1/4}$.

We now show that the events $B_i$ are independent. Condition on all Poisson counts and arrival times. Reveal the edge identities in chronological order. Each $e_{i,k}$ is determined before the $k$-th edge in $C_i$ is revealed, while that edge is independent and uniform in $E(C_i)$. Hence, by induction over the reveal order, the events that the $k$-th edge in $C_i$ equals $e_{i,k}$ are mutually independent over all $i,k$. Since the Poisson counts are also independent across cliques, so are the events $B_i$.

As there are \(M=\Omega(n^{63/64})\) cliques and the corresponding events $B_1,\ldots,B_M$ are independent, some \(B_i\) occurs with probability at least $1-(1-n^{-1/4}/2)^M = 1-o(1)$.
\end{proof}

\subsection{Missing Proof of \Cref{lem:rand-graph-lem}}
\label{subsec:structure-sampled}

The proof uses a ``witness-tree'' argument, which proceeds as follows. To bound
the number of edges in a connected component, we enumerate the multisets of
\(r\) edges of \(G\) that form a connected subgraph. For each such multiset,
we bound the probability that it occurs in \(\Gsamp\). A union bound then shows
that, with high probability, every connected multiset of edges occurring in
\(\Gsamp\) has size \(O(\log n)\). When \(\Delta\) is large, a different bound
on the number of connected multisets shows that every component is nearly a
tree.

\begin{proof}[Proof of \Cref{lem:rand-graph-lem}]
Set \(c:=16\), \(T:=n/c\), and \(M:=|E(G)|=n\Delta/2\). The \(T\) edges of
\(\Gsamp\) are sampled independently and uniformly with replacement from \(E(G)\).

\paragraph{Components have \(O(\log n)\) vertices and edges.}
We first show that, with high probability, every component of \(\Gsamp\) has
\(O(\log n)\) edges. The vertex bound then follows.

For a fixed \(r\ge1\), let \(C\) be an arbitrary multigraph formed by a
multiset of \(r\) edges of \(G\) whose underlying simple graph is connected.
To count all such \(C\), we construct an injective map
\(C\mapsto\mathtt{Witness}(C)\) to
vertex-labeled ``witness trees'' with \(r\) edges. It therefore suffices to
count the possible witness trees.

To construct \(\mathtt{Witness}(C)\), we choose a spanning tree
\(C_{\mathrm{tree}}\) of \(C\). For each remaining edge
\(\{u,v\}\in C\setminus C_{\mathrm{tree}}\), attach a new leaf labeled \(v\)
to the vertex \(u\) of \(C_{\mathrm{tree}}\). The result is a vertex-labeled
witness tree \(\mathtt{Witness}(C)\) with exactly \(r\) edges.\footnote{To make
the map well defined, fix once and for all a deterministic
rule for selecting \(C_{\mathrm{tree}}\) from the underlying simple graph of
\(C\), and another for deciding which endpoint of each remaining edge labels
the new leaf. The particular rules do not affect the argument.}
Note that its vertex labels need not be unique. Since the multiset of labeled edges in
\(\mathtt{Witness}(C)\) is exactly \(C\), the map is injective. The construction
is illustrated in \Cref{fig:witness-construction}.

\begin{figure}[ht]
\centering
\begin{tikzpicture}[
    vertex/.style={circle,draw,minimum size=6mm,inner sep=0pt},
    every node/.style={font=\small},
    line width=0.6pt
]
    \node[vertex] (a1) at (-0.8,-0.35) {$a$};
    \node[vertex] (b1) at (0.8,-0.35) {$b$};
    \node[vertex] (c1) at (0,0.9) {$c$};
    \draw (a1)--(c1)--(b1)--(a1);
    \draw (a1) to[bend right=35] (b1);
    \node at (0,-1.25) {$C$};

    \draw[-{Latex[length=2mm]}] (1.45,0.15)--(2.45,0.15);

    \node[vertex] (a2) at (3.1,0.45) {$a$};
    \node[vertex] (b2) at (4.1,0.45) {$b$};
    \node[vertex] (c2) at (5.1,0.45) {$c$};
    \node[vertex] (b3) at (3.1,-0.55) {$b$};
    \node[vertex] (a3) at (5.1,-0.55) {$a$};
    \draw (a2)--(b2)--(c2);
    \draw[dashed] (a2)--(b3);
    \draw[dashed] (c2)--(a3);
    \node at (4.1,-1.25) {$\mathtt{Witness}(C)$};
\end{tikzpicture}
\caption{
The selected spanning tree \(C_{\mathrm{tree}}\) is
\(a\text{--}b\text{--}c\). The two remaining edges \(\{a,b\}\) and
\(\{a,c\}\) are attached as dashed leaf edges to form
\(\mathtt{Witness}(C)\). The labels \(a\) and \(b\) therefore repeat.}
\label{fig:witness-construction}
\end{figure}
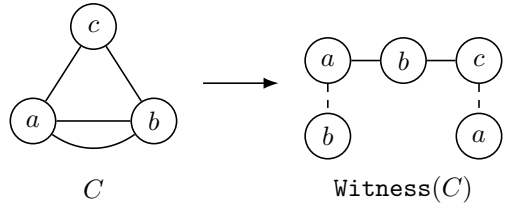

We now upper bound the number of witnesses. Each witness is a tree with \(r\) edges,
and each labeled edge belongs to \(E(G)\). There are at most \(4^r\) rooted
ordered tree shapes~\cite{Knuth1997TAOCP1}. We fix such a shape and enumerate the number of labelings that lead to a witness. The label of the root is chosen first. There are 
\(n\) possible ways to do this. Next, the remaining vertices are labeled in breadth-first order. Once the
label of a parent has been fixed, the label of each child can be chosen in at
most \(\Delta\) ways, since it must be a neighbor of its parent in \(G\). Thus,
\begin{align}
    \text{Number of witness trees}
    \le 4^r n\Delta^r.
    \label{eq:number-witnesses}
\end{align}

We next bound the probability that a fixed \(r\)-edge multiset \(S\) occurs in
\(\Gsamp\). Fix an ordering of the \(r\) edge occurrences in \(S\). There are at
most \(T^r\) choices of distinct sampling times at which these edges can be
sampled.
For any fixed choice, the probability that the corresponding edges are sampled
is \(M^{-r}\). Therefore,
\begin{align}
    \Pr[S\text{ occurs in }\Gsamp]
    \le T^rM^{-r}
    =\left(\frac{2}{c\Delta}\right)^r
    =\left(\frac{1}{8\Delta}\right)^r.
    \label{eq:fixed-edge-set}
\end{align}
Combining \eqref{eq:number-witnesses} and \eqref{eq:fixed-edge-set}, for each
\(r\), the probability that \(\Gsamp\) contains a connected \(r\)-edge multiset
is at most
\(4^rn\Delta^r(1/(8\Delta))^r=n/2^r\). Set
\(L:=\lceil10\log_2 n\rceil\). A union bound over
\(r=L,\ldots,T\) shows that the probability that some component has at least
\(L\) edges is at most
\(\sum_{r=L}^T n/2^r\le 2n/2^L\le 2n^{-9}\).
Thus, with probability at least \(1-2n^{-9}\), every component has
\(O(\log n)\) edges.

\paragraph{Components are ``tree-like'' for large \(\Delta\).}
Assume \(\Delta=\Omega((\log n)^3)\), and set
\(p:=\lceil 20\log n/\log\Delta\rceil\). By the first part, except with
probability at most \(2n^{-9}\), every component has at most \(L\) vertices.
Fix \(m\le L\), and let \(\mathcal B_m\) be the event that \(\Gsamp\) contains
a component with \(m\) vertices and at least \(m-1+p\) edge occurrences. We
will now bound the probability of \(\mathcal B_m\).

Consider a connected edge multiset \(C\) with \(m-1+p\) edges whose underlying
graph has \(m\) vertices. We again use the injective map
\(C\mapsto\mathtt{Witness}(C)\). It therefore suffices to count the possible
witnesses.

Each witness consists
of the tree \(C_{\mathrm{tree}}\), which has \(m-1\) edges and distinct vertex
labels, together with \(p\) new leaves attached to its vertices. Each new leaf
has one of the \(m\) labels already appearing in \(C_{\mathrm{tree}}\). Using
the same argument as above, there are at most \(4^m n\Delta^{m-1}\) choices for
\(C_{\mathrm{tree}}\). Once this tree is fixed, there are at most \(m^p\)
choices for the attachment points of the leaves and at most \(m^p\) choices for
their labels. Thus,
\begin{equation}
    \text{Number of witness trees}
    \le 4^m n\Delta^{m-1}m^{2p}.
    \label{eq:tree-like-witness-count}
\end{equation}
If \(\mathcal B_m\) occurs, then at least one such edge multiset \(C\) occurs in
\(\Gsamp\). Combining \eqref{eq:fixed-edge-set} and
\eqref{eq:tree-like-witness-count}, a union bound gives
\begin{align}
    \Pr[\mathcal B_m]
    &\le 4^m n\Delta^{m-1}m^{2p}
    \left(\frac{1}{8\Delta}\right)^{m-1+p}
    =\frac{8n}{2^m}\left(\frac{m^2}{8\Delta}\right)^p.
    \label{eq:event-bm}
\end{align}
Since \(m\le L=O(\log n)\) and \(\Delta=\Omega((\log n)^3)\), we have
\(m^2/(8\Delta)\le\Delta^{-1/4}\) for all sufficiently large \(n\). Since
\(p\ge 20\log n/\log\Delta\), we have
\(\Pr[\mathcal B_m]\le 8n\Delta^{-p/4}\le 8n^{-4}\). A union bound over
\(m\le L\) gives a total failure probability of \(O(\log n/n^4)\).

The first part fails with probability at most \(2n^{-9}\), and the second with
probability \(O(\log n/n^4)\). Thus, for all sufficiently large \(n\), with
probability at least \(1-1/n\), every component has \(O(\log n)\) vertices and
edges. Moreover, when \(\Delta=\Omega((\log n)^3)\), every component has at most
\(O(\log n/\log\Delta)\) edges beyond a spanning tree. This proves the lemma.
\end{proof}

\textbf{AI Disclosure:} An initial proof \Cref{thm:greedy}, which applied only to unit weight requests, was obtained by the authors. Upon prompting ChatGPT 5.5 with this proof, it produced a simplified proof of the same bound that also generalized to requests with weights in $[0,1]$. This is the version of the proof presented in Section 3 of the paper. The other proofs in the paper are by authors. AI tools were also used to assist with typesetting and proofreading. The authors take responsibility for the contents of the paper.

\bibliography{general}
\bibliographystyle{alpha}
\end{document}